\documentclass[12pt,oneside]{article}

\usepackage[round,colon,authoryear]{natbib} 

\usepackage{graphicx}
\usepackage{amsmath}
\usepackage{amsfonts}
\usepackage{amssymb}
\usepackage{amsthm}
\usepackage{enumerate}
\usepackage{times} 

\newcommand{\R}{\mathbb{R}}
\newcommand{\E}{\mathbb{E}}
\newcommand{\N}{\mathbb{N}}
\newcommand{\1}{\mathbf{1}}
\newcommand{\Prob}{\mathbb{P}}
\newcommand{\Q}{\mathbb{Q}}
\newcommand{\CF}{\mathcal{F}}
\newcommand{\tr}{^\mathsf{T}}

\newcommand{\Keywords}[1]{\par\noindent
{\small{\textsc KEY WORDS\/}: #1}}

\newtheorem{thm}{Theorem}
\newtheorem{prop}{Proposition}
\newtheorem{lemma}{Lemma}
\newtheorem{cor}{Corollary}

\theoremstyle{remark}
\newtheorem{ex}{Example}
\theoremstyle{remark}

\theoremstyle{remark}
\newtheorem{remark}{Remark}

\numberwithin{equation}{section}
\numberwithin{thm}{section}
\numberwithin{lemma}{section}
\numberwithin{cor}{section}
\numberwithin{remark}{section}
\numberwithin{ex}{section}
\numberwithin{defin}{section}

\allowdisplaybreaks

\title{HEDGING UNDER ARBITRAGE\thanks{I am indebted to two anonymous referees and an associate editor for their careful
reading and constructive feedback. Many thanks go to Ioannis
Karatzas and Hans F\"ollmer for sharing their insights and for
their  helpful comments on previous drafts of this work. I am
grateful to Michael Agne, Adrian Banner, Daniel Fernholz, Robert
Fernholz, Ashley Griffith, Tomoyuki Ichiba, Phi Long Nguen-Tranh, Sergio Pulido,
Subhankar Sadhukhan, Emilio Seijo, Li Song, Winslow Strong, Johan
Tysk, and Hao Xing for  fruitful discussions on the subject
matter of this paper. This work was partially supported by the
National Science Foundation DMS Grant
 09-05754 and by a Faculty Fellowship of Columbia University.
Results in this paper are drawn in part from the second chapter of the author's doctoral thesis \citet{Ruf_ots}
supervised by Ioannis Karatzas.
}}
\author{Johannes Ruf\footnote{Address correspondance to Johannes Ruf, Department of Statistics, Columbia University,
    1255 Amsterdam Avenue, New York, NY 10027, USA; ruf@stat.columbia.edu.} \\
    Columbia University \\ \today
    }
\date{}

\begin{document}
\thispagestyle{plain} \maketitle
\begin{abstract}
\noindent
   It is shown that delta hedging provides the optimal trading strategy in terms
    of minimal required initial capital
    to replicate a given terminal payoff in a continuous-time Markovian context.
    This holds true in market models in which no equivalent local martingale measure exists but only a
    square-integrable market price of risk.
    A new probability measure is constructed, which takes the place of an equivalent local martingale measure.
    In order to ensure the existence of the delta hedge, sufficient conditions are derived for the necessary differentiability of expectations indexed over
    the initial market configuration.
    The phenomenon of ``bubbles,'' which has recently been frequently discussed in the academic literature,
    is a special case of the setting in this paper.
    Several examples at the end illustrate the techniques described in this work.
    \Keywords{Benchmark Approach, Stochastic Portfolio Theory,
    bubbles,
    local martingales, F\"ollmer measure, continuous time, diffusions, stochastic discount factor, market price of risk,
    trading strategies, arbitrage, pricing,
    hedging, options, put-call-parity, Black-Scholes PDE, stochastic flows, Schauder estimates, Bessel process}
\end{abstract}
\section{INTRODUCTION}
In a financial market, an investor usually has several trading strategies at her disposal to obtain a given wealth at
a specified point in time.
 For example, if the investor wanted to cover a short-position in a given stock tomorrow at the cheapest cost today,
buying the stock today is generally not optimal, as there may
be a trading strategy requiring less initial capital that still
replicates the exact stock price tomorrow.  In this paper, we show that optimal trading strategies,
in the sense of minimal required initial capital, can be represented as delta hedges.

This paper has been motivated by the problem of finding trading strategies to exploit relative arbitrage
opportunities, which arise naturally in the framework of Stochastic Portfolio Theory (SPT).
For that, we generalize the results of
\citet{FK}'s paper ``On optimal arbitrage,'' in which specifically the market portfolio  is
examined, to a wide class of terminal wealths which can be optimally replicated by delta hedges.
For an overview of SPT and a discussion of relative arbitrage opportunities, we recommend the reader consult the
monograph by \citet{Fe} and the survey paper by \citet{FK_survey}.
The problem investigated here is directly linked to the question of computing hedges of contingent claims, which has
 been studied within the Benchmark Approach (BA), developed
by Eckhard Platen and co-authors.
Indeed, we generalize some of the results in the BA here and provide tools to compute the so-called
``real-world prices'' of contingent claims under that approach.
The monograph by \citet{PH} provides an excellent overview of the BA.

We shall not restrict ourselves only to markets satisfying the  the ``No free lunch with vanishing risk'' (NFLVR)
or, more precisely, the ``No arbitrage for general admissible integrands'' (NA) condition.\footnote{We
refer the reader to the monograph by \citet{DS_mono} for a thorough introduction to NA, NFLVR and other notions of arbitrage.
Since we shall assume the existence of a square-integrable market price of risk, we implicitly impose the condition that
NFLVR fails if and only if NA fails; see \citet{KK}, Proposition~3.2.}
Thus, we cannot rely on the existence of an equivalent local martingale measure (ELMM), which we otherwise would have
done.
However, we shall construct another probability
measure to take the place of the ``risk-neutral''
measure. We do not assume an ELMM a priori for several reasons. First, we cannot always assume the existence of a
statistical test that relies upon stock price observations to determine whether an ELMM exists, as illustrated in \citet{KK},
Example~3.7.  Second, examining arbitrage opportunities, rather than excluding them a priori, is of interest in itself.
Further arguments and empirical evidence supporting the consideration of models without an ELMM
 are discussed in \citet{K_balance}, Section~0.1 and \citet{PH_hedging}, Section~1.
A model of economic equilibrium for such models is provided in \citet{LW_snack}.
In the spirit of these papers, we shall impose some restrictions on the
arbitrage opportunities and exclude a priori models which imply ``unbounded profit with bounded risk,''
 which can be recognized by a typical agent.

There have been several recent papers treating the subject of ``bubbles'' within models guaranteeing NFLVR;
 a very incomplete list consists of the work
by \citet{LW_bubble}, \citet{CH}, \citet{HLW}, \citet{JP_complete, JP_incomplete}, \citet{PP}, and \citet{ET}.
A bubble is usually defined as the difference between the market price of a tradeable asset and its smallest hedging price.
The analysis here includes the case of bubbles, but is more
general, as it also allows for models without an ELMM.
To wit, while the bubbles literature concentrates on a single stock whose price process
is modeled as a strict local
martingale, we consider markets with several assets with the stochastic discount factor itself being
represented by a (possibly strict) local martingale.
In the case of an asset with a bubble, our contribution is limited to the explicit representation of the
optimal replicating strategy as a delta hedge. We shall also discuss in this context the reciprocal of the three-dimensional
Bessel process as the standard example for a bubble.

We set up our analysis in a continuous-time Markovian context; to wit, we focus on stock price processes whose
mean rates of return and volatility coefficients only depend on time and on the current market configuration.
Since we do not rely  on a martingale representation theorem, we can allow for a larger number of
driving Brownian motions than the number of stocks, which generalizes the ideas of \citet{FK} to
not only a larger set of payoffs, but also to a broader set of models for the specific case of the market portfolio.
We shall prove that a classical delta hedge yields the cheapest hedging strategy for European contingent claims.
This is of course well-known in the case where an ELMM exists and is extended here to models which
allow for arbitrage opportunities and that
are not necessarily complete. In this context, we provide sufficient conditions to ensure the differentiability
of the hedging price, generalizing results by \citet{HS}, \citet{JT}, and \citet{ET}.
This set of conditions is also applicable to models satisfying the NFLVR assumption.
Because the computations for the optimal trading strategy under the ``real-world'' measure are often too involved
and because we cannot always rely on an ELMM,
 we derive a non-equivalent change of measure including a generalized Bayes' rule.

The next section introduces the market model and trading strategies.
Section~\ref{S market price} provides a discussion about the market price of risk.
Section~\ref{S optimal strategies} contains the precise
representation of an optimal strategy to hedge a non path-dependent European
claim and sufficient conditions for the differentiability of the hedging price. A modified put-call parity follows directly.
We suggest in Section~\ref{S hedging price} a change to some non-equivalent probability measure that simplifies computations.
Section~\ref{S examples} then provides several examples  and
Section~\ref{S conclusion} draws the conclusions.
\section{MARKET MODEL AND TRADING STRATEGIES}  \label{S market}
In this section, we introduce the market model and trading
strategies. We assume the perspective of a small investor who
takes positions in a frictionless financial market with finite time
horizon $T$. We shall use the notation $\R_+^d := \{s = (s_1,
\ldots, s_d)\tr \in \R^d, s_i>0, \text{ for all } i = 1, \ldots,
d\}$ and assume a market in which the stock price processes are
modeled as positive continuous Markovian semimartingales. That
is, we consider a financial market $S(\cdot)=(S_1(\cdot), \ldots,
S_d(\cdot))\tr$ of the form
\begin{equation}  \label{market}
    d S_i(t) = S_i(t) \left( \mu_i (t, S(t)) d t +
    \sum_{k=1}^K \sigma_{i,k}(t, S(t)) d W_k(t) \right)
\end{equation}
for all $i=1,\ldots, d$ and $t \in [0,T]$ starting at $S(0) \in \R^d_+$ and a money market $B(\cdot)$.
Here $\mu: [0,T] \times \R^d_+ \rightarrow \R^d$ denotes the  mean rate of return and
$\sigma: [0,T] \times \R^d_+ \rightarrow \R^{d \times K}$ denotes the volatility.
We assume that both functions are measurable.

For the sake
of convenience we only consider discounted (forward) prices and
set the interest rate constant to zero; that is, $B(\cdot) \equiv 1$.
The flow of information is
modeled as a right-continuous filtration $\mathbb{F} = \{\CF(t)\}_{0 \leq t \leq T}$
such that $W(\cdot) = (W_1(\cdot), \ldots, W_K(\cdot))\tr$
is a $K$-dimensional Brownian motion with independent components. In Section \ref{S hedging price}, we
impose more conditions on the filtration $\mathbb{F}$ and the underlying probability space $\Omega$.
The underlying measure and its expectation will be denoted by $\Prob$ and $\E$, respectively.

We only
consider those mean rates of return $\mu$ and volatilities $\sigma$ that imply
the stock prices $ S_1 (\cdot), \cdots, S_d (\cdot)$ exist and are unique and strictly positive.
More precisely, denoting the  covariance process of
the stocks by
$a (\cdot, \cdot) = \sigma(\cdot, \cdot) \sigma\tr(\cdot, \cdot)$,
we impose the almost sure integrability condition
\begin{equation*}
  \sum_{i=1}^d \int_0^T  \left(\left|\mu_i (t, S(t)) \right| +  a_{i,i} (t, S(t)) \right)  d t  <  \infty.
\end{equation*}

Next, we introduce the notion of trading strategies and associated wealth processes to be able
to  describe formally delta hedging below.
We denote the number of shares held by an investor at time~$t$
 by $\eta(t) = (\eta_1(t), \ldots, \eta_d(t))^\mathsf{T}$ and  call $\eta(\cdot)$ a \emph{trading strategy}
or in short, a \emph{strategy}. We assume that $\eta(\cdot)$ is progressively measurable with respect to $\mathbb{F}$
and  self-financing. This yields for the
corresponding wealth process $V^{v,\eta}(\cdot)$ of an investor with initial capital $v>0$ the dynamics
\begin{equation*}
     d V^{v,\eta}(t) = \sum_{i=1}^d \eta_i(t) dS_i(t)
\end{equation*}
for all $t \in [0,T]$ and $V^{v,\eta}(0)=v$. To ensure that
$V^{v,\eta}(\cdot)$ is well-defined and to exclude doubling
strategies we restrict ourselves to trading strategies which
satisfy $V^{v,\eta}(t) \geq 0$ for a given initial wealth $v>0$, and the almost sure integrability
condition
    \begin{equation*}
        \sum_{i=1}^d \int_0^T  \left(
        S_i(t) |\eta_i(t) \mu_i(t, S(t))| + S_i^2(t) \eta_i^2(t) a_{i,i}(t, S(t))\right) d t < \infty.  \label{pi int cond}
    \end{equation*}
\section{MARKET PRICE OF RISK AND STOCHASTIC DISCOUNT FACTOR} \label{S market price}
This section discusses two important components of the market model.
We assume that the market model of \eqref{market} implies a \emph{market price of risk (MPR)},
which generalizes the concept of the Sharpe ratio
to several dimensions.
More precisely, an MPR is a progressively measurable process $\theta(\cdot)$,
     which maps the volatility structure $\sigma$ onto the mean rate of return $\mu$. That is,
    \begin{equation} \label{E market price}
        \mu(t,S(t)) = \sigma(t,S(t)) \theta(t)
    \end{equation}
    for all $t \in [0,T]$ holds almost surely.
We further assume that $\theta(\cdot)$ is square-integrable, to
wit,
    \begin{equation} \label{E market price int}
        \int_0^T \|\theta(t)\|^2 dt < \infty
    \end{equation}
almost surely.
An MPR does not have to be
uniquely determined. Uniqueness is intrinsically connected to
completeness, which we need not assume. In general, infinitely
many MPRs may exist. An example for non-uniqueness
is given following Proposition~\ref{P independence} below.

The existence of an MPR is a central assumption in both the BA
\citep[see][Chapter~10]{PH} and SPT \citep[see][Section~6]{FK_survey}.
This assumption enables us to discuss hedging prices, as we do throughout this paper, since it
excludes scalable arbitrage opportunities by guaranteeing
`no unbounded profit with bounded risk'' (NUPBR) as demonstrated in \citet{KK}.
Similar assumptions
    have been discussed in the economic literature. For example, in the terminology of \citet{LW_snack},  the existence
    of a square-integrable MPR excludes ``cheap thrills''
    but not necessarily ``free snacks.'' Theorem~2 of \citet{LW_snack} shows that a market with a square-integrable MPR
    is consistent with an equilibrium where agents prefer more to less.

Based upon the MPR, we can now define the \emph{stochastic discount factor (SDF)} as
\begin{equation}
 \label{E Z}
      Z^\theta(t) :=   \exp \left( - \int_0^t
   \theta\tr (u)  d W(u) - \frac{1}{2} \int_0^t
   \| \theta (u)\|^2  du \right)
\end{equation}
for all $t \in [0,T]$.
In classical no-arbitrage theory, $Z^\theta(\cdot)$ represents the Radon-Nikodym derivative
which translates the ``real-world'' measure into the generic ``risk-neutral'' measure with the money market as the underlying.
Since we do not want to impose NFLVR a priori in this work, but are rather interested in situations in which NFLVR does not
necessarily hold, we shall not assume that the SDF $Z^\theta(\cdot)$ is a true martingale.
Cases where $Z^\theta(\cdot)$ is only a local martingale have, for example, been discussed by
\citet{KLSX}, \citet{Schweizer_1992},
in the BA starting with \citet{Pl2002} and \citet{HP_consistent, HP_index} and in SPT;
see, for example, \citet{FKK} and especially, \citet{FK}.

In this context, it is important to remind ourselves that $Z^\theta(\cdot)$ is a true martingale if and only if
there exists an ELMM $\Q$, under which the stock price processes are local martingales.
The question of whether
$\Q$ is a martingale measure or only a local martingale measure is not connected to whether $Z^\theta(\cdot)$ is a strict local
or a true martingale.  A \emph{bubble} is usually defined within a model in which $Z^\theta(\cdot)$ is a true martingale. Then,
a wealth process is said to have a bubble if it is a strict local martingale under an ELMM.\footnote{In the bubbles literature, there has
been an alternative definition, based upon the characterization of the pricing operator as a
finitely additive measure. It can be shown that this characterization is equivalent to the one here;
see \citet{JP_incomplete}, Section~8 for the proof and literature which relies on
this alternative characterization.}
\citet{JP_complete, JP_incomplete}  suggest replacing the NFLVR condition by
 the stronger condition of ``no dominance'' first proposed by \citet{Merton} to exclude bubbles.
Here, we take the opposite approach. Instead of imposing a new condition, the goal of this analysis is to investigate
a general class of models and study how much can be said in this more general
framework without having the tool of an ELMM.

We observe that the existence of a square-integrable MPR implies the
existence of a Markovian square-integrable MPR. To see this, we
define $\theta(\cdot, \cdot) := \sigma\tr(\cdot, \cdot)
(\sigma(\cdot, \cdot)\sigma\tr(\cdot, \cdot))^\dagger \mu(\cdot, \cdot)$, where $\dagger$ denotes the Moore-Penrose pseudo-inverse of a matrix.
Given the existence of any MPR, we know from the theory of least-squares estimation
that $\theta(\cdot, \cdot)$ is also a MPR. Furthermore, we have $\|\theta(t,S(t))\|^2 \leq \|\nu(t)\|^2$ for all
$t \in [0,T]$ almost surely for any MPR $\nu(\cdot)$, which yields the square-integrability of $\theta(\cdot, \cdot)$.
This observation has been pointed out to us by a referee.

The next proposition shows that any square-integrable Markovian MPR
maximizes the random variable which will later be a candidate for
a hedging price. We denote by $\mathcal{F}^S(\cdot)$ the augmented
filtration generated by the stock price process. We emphasize that the next result only holds so long as the ``terminal
payoff'' $M$ is $\mathcal{F}^S(T)$-measurable.
\begin{prop}[Role of Markovian MPR] \label{P independence}
    Let $M \geq 0$ be a random variable
    measurable with respect to $\mathcal{F}^S(T) \subset \mathcal{F}(T)$. Let $\nu(\cdot)$ denote any square-integrable MPR and
    $\theta(\cdot, \cdot)$ any
    Markovian square-integrable MPR. Then, with
    \begin{align*}
        M^{\nu}(t) := \E\left[\left.\frac{Z^{\nu}(T)}{Z^{\nu}(t)} M \right|\mathcal{F}_t\right]
        \text{ and }
       M^{\theta}(t) := \E\left[\left.\frac{Z^{\theta}(T)}{Z^{\theta}(t)} M \right|\mathcal{F}_t\right]
    \end{align*}
    for $t \in [0,T]$, where we take the right-continuous modification\footnote{See \citet{KS1}, Theorem~1.3.13.} for each process, we have
    $M^\nu(\cdot) \leq M^\theta(\cdot)$ almost surely.  Furthermore, if both $Z^\nu(\cdot)$ and $Z^\theta(\cdot)$
    are $\mathcal{F}^S(T)$-measurable, then $Z^\nu(T) \leq Z^\theta(T)$ almost surely.
\end{prop}
\begin{proof}
    Due to the right-continuity of $M^{\nu}(\cdot)$ and $M^{\theta}(\cdot)$  it suffices to show for all $t \in [0,T]$ that
    $M^\nu(t) \leq M^\theta(t)$ almost surely.
    We define  $c(\cdot) := \nu(\cdot)  - \theta(\cdot, S(\cdot))$.  For the sequence of stopping times
    \begin{align*}
         \tau_n := T \wedge \inf \left\{t \in [0,T]: \int_0^t c^2(s) ds \geq n\right\},
    \end{align*}
    where $n\in \N$, we set $c^n(\cdot) := c(\cdot) \1_{\{\tau_n \geq \cdot\}}$
    and observe that
    \begin{align*}
       \frac{Z^\nu(T)}{Z^\nu(t)}  &= \frac{Z^{c}(T)}{Z^{c}(t)}
            \cdot \exp\left(-\int_t^T \theta\tr(u, S(u)) (dW(u) + c(u) du) - \frac{1}{2} \int_t^T \|\theta(u, S(u))\|^2 du\right)\\
        &= \lim_{n \rightarrow \infty} \frac{Z^{c^n}(T)}{Z^{c^n}(t)}
            \cdot \exp\left(-\int_t^T \theta\tr(u, S(u)) (dW(u) + c^n(u) du) - \frac{1}{2} \int_t^T \|\theta(u, S(u))\|^2 du\right)
    \end{align*}
    with $Z^{c}(\cdot)$ and $Z^{c^n}(\cdot)$ defined as in \eqref{E Z}.
    The limit holds almost surely since both $v(\cdot)$ and $\theta(\cdot, \cdot)$ are square-integrable, which again
    yields the square-integrability of $c(\cdot)$.
    Since $\int_0^T {c^{n^2}}(t) dt \leq n$, Novikov's Condition \citep[see][Proposition~3.5.12]{KS1} yields that
    $Z^{c^n}(\cdot)$ is a martingale.
    Now, Fatou's lemma, Girsanov's theorem and
    Bayes' rule \citep[see][Chapter~3.5]{KS1} yield
    \begin{equation}  \label{E inequality M}
       M^\nu(t) \leq \liminf_{n \rightarrow \infty} \E^{\Q^n}\left[\left.
            \exp\left(-\int_t^T \theta\tr(u, S(u)) dW^n(u) - \frac{1}{2} \int_t^T \|\theta(u, S(u))\|^2 du\right) M \right|\mathcal{F}_t\right],
    \end{equation}
    where $d\Q^n(\cdot) := Z^{c^n}(T) d\Prob(\cdot)$ is a probability measure, $\E^{\Q^n}$ its
    expectation operator,
     and $W^n(\cdot) := W(\cdot) + \int_0^\cdot c^n(u) du$ a $K$-dimensional $\Q^n$-Brownian motion.
    Since $\sigma(\cdot,S(\cdot)) c^n(\cdot) \equiv 0$ we can replace $W(\cdot)$ by $W^n(\cdot)$ in
    \eqref{market}. This yields that the process $S(\cdot)$ has the same dynamics under $\Q^n$ as under $\Prob$.
    Furthermore, both $\theta(\cdot, S(\cdot))$ and $M$ have, as functionals of $S(\cdot)$, the same
    distribution under  $\Q^n$ as under $\Prob$.
    Therefore, we can replace the expectation operator $\E^{\Q^n}$ by $\E$ and the Brownian motion $W^n(\cdot)$ by $W(\cdot)$
    in \eqref{E inequality M} and
    obtain the first part of the statement.  The last inequality of the statement follows from setting
    $M = \1_{\{Z^\nu(T) > Z^\theta(T)\}}$ and observing that $M$ must equal zero almost surely.
\end{proof}
We remark that  the
inequality $M^\nu(\cdot) \leq M^\theta(\cdot)$ can be strict.  For  an example, choose $M = 1$ and a market with
    one stock and two Brownian motions, to wit,
    $d=1$ and $K=2$. We set $\mu(\cdot, \cdot) \equiv 0$, $\sigma(\cdot, \cdot) \equiv (1, 0)$ and observe
    that $\theta(\cdot, S(\cdot)) \equiv (0, 0)\tr$ is a
    Markovian MPR. Another MPR $\nu(\cdot) \equiv (\nu_1(\cdot), \nu_2(\cdot))\tr$ is defined
    via $\nu_1(\cdot) \equiv 0$, the stochastic differential equation
    $
        d\nu_2(t) = -\nu_2^2(t) dW_2(t)
    $
    for all $t \in [0,T]$ and $\nu_2(0)=1$. That is, $\nu_2(\cdot)$ is the reciprocal of a three-dimensional Bessel process
    starting at one. Since $Z^\nu(\cdot)$ also satisfies the stochastic differential equation $dZ^\nu(t) = -Z^\nu(t) \nu_2(t) dW_2(t)$
    we have from \citet{JacodS}, Theorem~1.4.61 that $Z^\nu(\cdot) \equiv \nu_2(\cdot)$, which is a strict local martingale \citep[see][Exercise~3.3.36]{KS1},
    and thus
    $M^\nu(0) = \E[Z^\nu(T)] < 1 = \E[Z^\theta(T)] = M^\theta(0)$.

Under the assumption that an ELMM exists, \citet{Jacka}, Theorem~12, \citet{AS_couverture}, Theorem~3.2
or \citet{DS_numeraire}, Theorem~16 show that a contingent claim can be hedged if and only if the supremum over all
expectations of the terminal value of the contingent claim under all ELMMs is a maximum.
In our setup, we also observe that the supremum over all $M^{\tilde{\nu}}(0)$ in the last proposition is a maximum,
attained by  any
Markovian MPR. Indeed, we will prove in Theorem~\ref{T Respresentation markovian} that, under weak
analytic assumptions, claims of the form $M = p(S(T))$ can be hedged. The general theory lets us conjecture that all
claims measurable with respect to $\CF^S(T)$ can be hedged.

As pointed out by Ioannis Karatzas in a personal communication (2010), Proposition~\ref{P independence} might be related to the
``Markovian selection results,'' as in \citet{Krylov_Markov_Selection}, \citet{EK_Markovian}, Section~4.5, and
\citet{SV_multi}, Chapter~12.
There, the existence of a Markovian solution for a martingale problem is studied.
It is observed that a supremum over a set of expectations indexed by a
family of distributions is attained and the maximizing distribution is a Markovian solution of the martingale problem.
This potential connection needs to be worked out in a future research project.

From this point forward, we shall always assume the MPR to be Markovian. As we shall see, this choice will
lead directly to the optimal trading strategy.
\section{OPTIMAL STRATEGIES} \label{S optimal strategies}
In this section, we show that delta hedging provides the optimal trading strategy in terms
    of minimal required initial capital
    to replicate a given terminal payoff. Next, we prove a modified put-call parity.
    In order to ensure the existence of the delta hedge,
    we derive sufficient conditions for the differentiability of expectations indexed over
    the initial market configuration.

We will rely on the following notation.
If $Y$ is a nonnegative $\CF(T)$-measurable random variable such that $\E[Y|\CF(t)]$ is a function of $t$ and $S(t)$ for
all $t \in [0,T]$, we use the Markovian structure of $S(\cdot)$ to denote conditioning
on the event $\{S(t) = s\}$ by $\E^{t,s}[Y]$.
Outside of the expectation operator we denote
by $(S^{t,s}(u))_{u \in [t,T]}$ a stock price process with the dynamics of \eqref{market} and
$S(t) = s$, in particular, $S^{0,S(0)}(\cdot) \equiv S(\cdot)$.
 We observe that $Z^\theta(u)/Z^\theta(t)$  depends for $u \in (t, T]$ on $\mathcal{F}(t)$ only through
$S(t)$ and we write similarly $(\tilde{Z}^{\theta, t,s}(u))_{u \in [t,T]}$ for
$(Z^\theta(u)/Z^\theta(t))_{u \in [t,T]}$ with $\tilde{Z}^{\theta, t,s}(t) = 1$
on the event $\{S(t) = s\}$.  When we want to stress the dependence of a process on the
state $\omega \in \Omega$ we will write, for example, $S(t,\omega)$.

Let us denote by $\text{supp}(S(\cdot))$ the support of $S(\cdot)$, that is, the smallest closed set in $[0,T] \times \R^n$ such that
\begin{align*}
    \Prob( (t,S(t)) \in \text{supp}(S(\cdot)) \text{ for all } t \in [0,T]) = 1.
\end{align*}
We call $\text{i-supp}(S(\cdot))$ the union of $(0,S(0))$ and the interior of $\text{supp}(S(\cdot))$ and assume that
\begin{align*}
    \Prob( (t,S(t)) \in \text{i-supp}(S(\cdot)) \text{ for all } t \in [0,T)) = 1.
\end{align*}
This assumption is made to exclude degenerate cases, where $S(\cdot)$ can hit the boundary of its support with positive probability.
We shall call any $(t,s) \in \text{i-supp}(S(\cdot))$ a \emph{point of support for $S(\cdot)$} and we remark that
each such point $(t,s)$ satisfies $t < T$.  For example, if $S(\cdot)$ is a one-dimensional geometric Brownian motion
then the set of points of support for $S(\cdot)$ is exactly $(0,S(0)) \cup \{(t,s) \in (0,T) \times \R_+\}$.


We define for any measurable function $p: \R^d_+ \rightarrow [0,\infty)$ a candidate
$h^p: [0,T] \times \R^d_+ \rightarrow [0,\infty)$ for the
hedging price of the corresponding European option:
\begin{equation}  \label{D g}
    h^p(t,s) := \E^{t,s}\left[\tilde{Z}^{\theta, t,s}(T) p(S(T))\right].
\end{equation}
Since $S(\cdot)$ is Markovian, $h^p$ is well-defined.
Proposition~\ref{P independence} yields that $h^p$ does not depend on the choice of
the (Markovian) MPR $\theta(\cdot)$.
Equation~\eqref{D g} has appeared
as the ``real-world pricing formula'' in the BA; compare \citet{PH}, Equation~(9.1.30).
Simple examples for payoffs could be the market portfolio
($\tilde{p}(s) = \sum_{i=1}^d s_i$), the money market ($p^0(s) = 1$), a stock ($p^1(s) = s_1$), or a call
($p^C(s) = (s_1 - L)^+$ for some $L \in \R$). We can now prove the first main result, which in particular provides
a mechanism for pricing and hedging contingent claims under the BA.
We denote by $D_i$, $D^2_{i,j}$ the partial derivatives with respect to the variable $s$.
\begin{thm}[Markovian representation for non path-dependent European claims] \label{T Respresentation markovian}
    Assume that we have a contingent claim of the form
    $p(S(T)) \geq 0$ and that the function $h^p$ of \eqref{D g} is sufficiently
    differentiable or, more precisely, that for all points of support $(t,s)$ for $S(\cdot)$ we have $h^p \in
    C^{1,2}(\mathcal{U}_{t,s})$ for some neighborhood $\mathcal{U}_{t,s}$ of $(t,s)$.
    Then, with
    \begin{equation*}
        \eta^p_i(t,s) := D_i h^p(t,s)
    \end{equation*}
    for all $i = 1, \ldots, d$ and $(t,s) \in [0,T] \times \R^d_+$,
    and with $v^p := h^p(0,S(0))$,
    we get
    \begin{equation*}
        V^{v^p,\eta^p}(t) = h^p(t,S(t))
    \end{equation*}
    for all $t \in [0,T]$.
    The strategy $\eta^p$ is optimal in the sense that for any $\tilde{v} > 0$ and for any strategy
    $\tilde{\eta}$ whose associated wealth process is nonnegative and satisfies
    $V^{\tilde{v}, \tilde{\eta}}(T) \geq p(S(T))$ almost surely,
    we have $\tilde{v} \geq v^p$.
    Furthermore, $h^p$ solves the PDE
    \begin{equation} \label{E PDE h}
        \frac{\partial}{\partial t} h^p(t,s) + \frac{1}{2} \sum_{i=1}^d \sum_{j=1}^d s_i s_j a_{i,j}(t,s)
            D_{i,j}^2 h^p(t,s) = 0
    \end{equation}
    at all points of support $(t,s)$ for $S(\cdot)$.
\end{thm}
\begin{proof}
    Let us start by defining the martingale $N^p(\cdot)$ as
    \begin{equation*}
        N^p(t) :=  \E[Z^\theta(T) p(S(T))|\CF(t)] =  Z^\theta(t) h^p(t,S(t))
    \end{equation*}
    for all $t \in [0,T]$.
    Although $h^p$ is not assumed to be in $C^{1,2}([0,T) \times \R^d)$ but only to be locally smooth,
    we can apply a localized version of
    It\^o's formula \citep[see for example][Section~IV.3]{RY}  to it.
    Then, the product rule of stochastic calculus can be used to obtain the dynamics of $N^p(\cdot)$.
    Since $N^p(\cdot)$ is a martingale, the corresponding $dt$ term must disappear.
    This observation, in connection with \eqref{E market price} and the positivity of $Z^\theta(\cdot)$,
    yields PDE~\eqref{E PDE h}.
    It\^o's formula, now applied to $h^p(\cdot,S(\cdot))$, and PDE~\eqref{E PDE h} imply
    \begin{align*}
        d h^p(t,S(t)) = \sum_{i=1}^d D_i h^p(t,S(t)) dS_i(t) = d V^{v^p,\eta^p}(t)
    \end{align*}
    for all $t \in [0,T]$.  This yields directly $V^{v^p,\eta^p}(\cdot) \equiv h^p(\cdot, S(\cdot))$.

    Next, we prove optimality.
    Assume we have some initial wealth $\tilde{v}>0$ and some strategy $\tilde{\eta}$ with nonnegative
    associated wealth process such that
    $V^{\tilde{v}, \tilde{\eta}}(T) \geq p(S(T))$ is satisfied almost surely. Then,
    $Z^\theta(\cdot) V^{\tilde{v}, \tilde{\eta}}(\cdot)$ is
    bounded from below by zero, thus a supermartingale. This implies
    \begin{equation*}
        \tilde{v} \geq \E[Z^\theta(T) V^{\tilde{v},\tilde{\eta}}(T)] \geq
         \E[Z^\theta(T) p(S(T))] = \E[Z^\theta(T) V^{v^p,\eta^p}(T)]  = v^p,
    \end{equation*}
    which concludes the proof.
\end{proof}
The last result generalizes
\citet{PH_hedging}, Proposition~3, where the same statement has been shown for a one-dimensional,
complete market with a time-transformed squared Bessel process of dimension four modeling the stock price process.
There are usually several strategies to obtain the same payoff. For example, if the first stock has a bubble, that is,
if $\E[Z^\theta(T) S_1(T)] < S_1(0)$, then one could either delta hedge with initial capital $\E[Z^\theta(T) S_1(T)]$
as the last theorem describes, or hold the stock with initial capital $S_1(0)$.
The last result shows that the delta hedge is optimal in the sense of minimal required
initial capital.
\citet{P_minimal} has suggested calling the fact that an optimal strategy exists the ``Law of the
Minimal Price'' to contrast it to the classical ``Law of the One Price,'' which appears if there is an equivalent
martingale measure.

We would like to emphasize that we have not shown that $\eta^p$ is unique. Indeed, since we have not
excluded the case that two stock prices have identical dynamics this is not necessarily true.
The next remark discusses the fact that we have not assumed the completeness of the market.
\begin{remark}[Completeness of the market]
One remarkable feature of the last theorem is that it does not require the market to be complete.
In particular, at no point have we assumed invertibility or full rank of the volatility matrix $\sigma(\cdot, \cdot)$.
In contrast to
\citet{FK}, we do not rely on the martingale representation theorem here but instead directly derive
a representation for the conditional expectation process of the final wealth $p(S(T))$.
The explanation for this phenomenon is that all relevant sources of risk for hedging are
completely captured by the tradeable stocks.
However, we remind the reader that we live here in a setting in which
the mean rates of return and volatilities do not depend on an extra stochastic factor.
In a ``more incomplete'' model, with jumps or additional risk factors in mean rates of return or
volatilities, this result can no longer be expected to hold.
Furthermore, there is no hope  to be able to hedge all contingent claims of the Brownian motion
$W(T)$. However, $W(T)$ appears in the model only as a nuisance parameter and it is of no economic interest to trade
in it directly.
\end{remark}
In the next remark we discuss PDE~\eqref{E PDE h}.
\begin{remark}[Non-uniqueness of PDE~\eqref{E PDE h}]  \label{R non-uniqueness}
Parabolic PDEs generally do not have unique solutions.
The hedging price for the stock of Example~\ref{Inverse Bessel} in \eqref{E h price iB},
for instance, is one of many
solutions  of polynomial growth for the corresponding Black-Scholes type PDE with terminal condition $p(s) = s$
and  boundary condition $f(t) = 0$.  Another solution is of course $h(t,s) = s$.
The reason for non-uniqueness in this case is the fact that the second-order
coefficient has super-quadratic growth preventing standard theory cannot from being applied; see, for example,
\citet{KS1}, Section~5.7.B.  However, one can show easily that, given that $h^p$ is sufficiently differentiable,
$h^p$ can be characterized as the minimal nonnegative classical solution of PDE~\eqref{E PDE h} with terminal condition
$h^p(T,s) = p(s)$; compare the proof of \citet{FK}, Theorem~1.
\end{remark}

\citet{FKK}, Example~9.2.2
illustrates that the classical put-call parity can fail. However, a modified version holds. An equivalent version
for the situation of an ELMM with possible bubbles has already
been found in \citet{JP_complete}, Lemma~7.
\begin{cor}[Modified put-call parity] \label{C put call}
    For any $L \in \R$ we have the modified put-call parity for the call- and put-options $(S_1(T) - L)^+$
    and $(L - S_1(T))^+$, respectively, with strike price $L$:
    \begin{equation}  \label{E put call}
        \E^{t,s}\left[\tilde{Z}^{\theta, t,s}(T)(L - S_1(T))^+\right] + h^{p^1}(t,s)  =
            \E^{t,s}\left[\tilde{Z}^{\theta, t,s}(T)(S_1(T) - L)^+\right] + L h^{p^0}(t,s),
    \end{equation}
    where $p^0(\cdot) \equiv 1$ denotes the payoff of one monetary unit and
    $p^1(s) = s_1$ the price of the first stock for all $s \in \R^d_+$.
\end{cor}
\begin{proof}
    The statement follows from the linearity of expectation.
\end{proof}
    Due to Theorem~\ref{T Respresentation markovian}, there exist, under weak differentiability assumptions,
    optimal strategies for the money market,
    the stock $S_1(T)$, the call and the put.
    Thus, the left-hand side of \eqref{E put call} corresponds
    to the sum of the hedging prices of a put and the stock, and the right-hand side corresponds to the
    sum of the hedging prices of a call
    and  $L$ monetary units.
    The difference between this and the classical put-call parity is that the current stock price and the strike $L$
    are replaced by their hedging prices.
    \citet{BKX}, Section~2.2 have recently observed another version. Instead of replacing the current stock
    price by its hedging price, they replace the European call price by the American call price and restore the
    put-call parity this way.

Next, we will provide sufficient conditions under which the function $h^p$ is sufficiently smooth. We shall call
 a function $f:[0,T] \times \R^d_+ \rightarrow \R$ \emph{locally Lipschitz and locally bounded} on $\R^d_+$ if for all
  $s \in \R^d_+$
    the function $t \rightarrow f(t,s)$ is right-continuous with left limits and for all $M>0$ there exists some $C(M) < \infty$
    such that
    \begin{equation*}
        \sup_{\substack{\frac{1}{M} \leq \|y\|, \|z\| \leq M \\ y \neq z}} \frac{|f(t,y) - f(t,z)|}{\|y-z\|} +
            \sup_{\frac{1}{M} \leq \|y\| \leq M} |f(t,y)| \leq C(M)
    \end{equation*}
    for all $t \in [0,T]$.
In particular, if $f$ has continuous partial derivatives, it is locally Lipschitz and locally bounded.
We require several
assumptions in order to show the differentiability of $h^p$ in Theorem~\ref{T diff} below.
\begin{itemize}
  \setlength{\itemsep}{1pt}
  \setlength{\parskip}{0pt}
  \setlength{\parsep}{0pt}
    \item[(A1)] The functions $\theta_k$ and $\sigma_{i,k}$ are for all
    $i = 1, \ldots, d$ and $k = 1, \ldots, K$
    locally Lipschitz and locally bounded.
    \item[(A2)] For all points of support $(t,s)$ for $S(\cdot)$
    there exist some $C>0$ and some neighborhood $\mathcal{U}$ of $(t,s)$  such that
    \begin{equation}  \label{E ineq A2}
        \sum_{i = 1}^d \sum_{j = 1}^d a_{i,j}(u,y) \xi_i \xi_j \geq C \|\xi\|^2
    \end{equation}
    for all $\xi \in \R^d$ and $(u,y) \in \mathcal{U}$.
    \item[(A3)] The payoff function $p$ is chosen so that for all points of support $(t,s)$ for $S(\cdot)$
        there exist some $C>0$ and some neighborhood $\mathcal{U}$ of $(t,s)$  such that $h^p(u,y) \leq C$
        for all $(u,y) \in \mathcal{U}$.
\end{itemize}
If $h^p$ is constant for $\tilde{d} \leq d$ coordinates, say the last ones, Assumption~(A2) can be weakened to
requesting the uniform ellipticity only in the remaining $d-\tilde{d}-1$ coordinates; that is, the sum in
\eqref{E ineq A2} goes only to $d-\tilde{d}-1$ and $\xi \in \R^{d-\tilde{d}-1}$.
Assumption~(A3) holds in particular if $p$ is of linear growth; that is, if
        $p(s) \leq C \sum_{i=1}^d s_i$ for
        some $C>0$ and all $s \in \R^d_+$,    since
        $\tilde{Z}^{\theta,t,s}(\cdot) S_i^{t,s}(\cdot)$ is a nonnegative supermartingale for all $i = 1, \ldots, d$.
We emphasize that the conditions here are weaker than the ones by \citet{FK}, Section~9 for the case of the market portfolio which
can be represented as $p(s) = \sum_{i=1}^d s_i$. In particular, the stochastic integral component in $Z^\theta(\cdot)$
does not present any technical difficulty in our approach.

We proceed in two steps. In the first step we use the theory of stochastic flows to derive continuity of $S^{t,s}(T)$
and $\tilde{Z}^{\phi,t,s}(T)$ in $t$ and $s$.  This theory relies on Kolmogorov's lemma, see, for example,
\citet{Protter}, Theorem~IV.73, and studies continuity of stochastic processes as functions of their initial
conditions. We refer the reader to \citet{Protter}, Chapter~V for
an introduction to and further references for stochastic flows.  We will prove continuity of $S^{t,s}(\cdot)$
and $\tilde{Z}^{\phi,t,s}(\cdot)$ at once and introduce for that the $d+1$-dimensional process
 $X^{t,s, z}(\cdot) := ({S^{t,s}}\tr(\cdot), z \tilde{Z}^{\phi,t,s}(\cdot))\tr$.
\begin{lemma}[Stochastic flow] \label{L flow}
    We fix a point $(t,s) \in [0,T] \times \R^d_+$ so that
    $X^{t,s, 1}(\cdot)$ is strictly positive and
    an $\R^{d+1}_+$-valued process. Then under Assumption~(A1) we have for all sequences $(t_k, s_k)_{k \in \N}
    \subset [0,T] \times \R^d_+$ with
    $\lim_{k \rightarrow \infty} (t_k, s_k) = (t,s)$ that
    \begin{equation*}
        \lim_{k \rightarrow \infty} \sup_{u \in [t,T]} \|X^{t_k,s_k,1}(u) - X^{t,s,1}(u)\| = 0
    \end{equation*}
    almost surely, where we set $X^{t_k,s_k,1}(u) := (s_k\tr,1)\tr$ for $u \leq t_k$.
    In particular, for $K(\omega)$ sufficiently large
    we have that
    $X^{t_k,s_k,1}(u,\omega)$ is strictly positive and $\R^{d+1}_+$-valued for all $k > K(\omega)$ and $u \in [t,T]$.
\end{lemma}
\begin{proof}
    Since the class of locally Lipschitz and locally bounded functions is closed under summation and multiplication,
    Assumption~(A1) yields that the drift and diffusion coefficients of
    $X^{u,y,z}(\cdot)$ are locally Lipschitz for all $(u,y,z) \in [0,T] \times \R^d_+ \times \R_+$. We start by
    assuming $t_k \geq t$ for all $k \in \N$ and obtain
    \begin{align}
        \sup_{u \in [t,T]} \|X^{t_k,s_k,1}(u) - X^{t,s,1}(u)\| \leq&
            \sup_{u \in [t,t_k]} \|(s_k\tr,1)\tr - X^{t,s,1}(u)\| + \sup_{u \in [t_k,T]} \|X^{t_k,s_k,1}(u) - X^{t_k,s,1}(u)\| \nonumber \\
            &+ \sup_{u \in [t_k,T]} \|X^{t_k,s,1}(u) - X^{t_k,S^{t,s}(t_k), \tilde{Z}^{\phi,t,s}(t_k)}(u)\|  \label{E inequality stoch flow}
    \end{align}
    for all $k \in \N$. The first term on the right-hand side of the last inequality goes to zero as $k$ increases
    by the continuity of the sample paths of $X^{t,s,1}(\cdot)$. The arguments in the proof of \citet{Protter}, Theorem~V.38 yield
    that
    \begin{equation*}
        \lim_{k \rightarrow \infty} \sup_{u \in [\tilde{t},T]} \|X^{\tilde{t},y_k,z_k}(u) - X^{\tilde{t},s,1}(u)\| = 0
    \end{equation*}
    for all $\tilde{t} \in \{t, t_1, t_2, \ldots\}$ and any sequence
    $((y_k\tr,z_k)\tr)_{k \in \N} \subset \R^{d+1}_+$ with $(y_k\tr,z_k)\tr \rightarrow (s\tr,1)\tr$ as $k \rightarrow \infty$
    almost surely.  An analysis of the arguments in \citet{Protter}, Theorems~V.37 and IV.73
    yields that the convergence is uniformly in $\tilde{t} \in \{t, t_1, t_2, \ldots\}$, see also
    \citet{Ruf_ots}, Lemma~1.  We now choose for
    $(y_k\tr,z_k)\tr$ the sequences $(s_k\tr,1)\tr$ and $({S^{t,s}}\tr(t_k,\omega),\tilde{Z}^{\phi,t,s}(t_k,\omega))\tr$
    for all $\omega \in \Omega$.
    This proves the statement if
    $t_k \geq t$ for all $k \in \N$. In the case of the reversed inequality $t_k \leq t$, a small modification of the inequality in
    \eqref{E inequality stoch flow} yields the lemma.
\end{proof}
In the second step, we use techniques from the theory of PDEs to conclude the necessary smoothness of
$h^p$. The following result has been used by Ekstr\"om, Janson and Tysk. We present it here on its own
to underscore the analytic component of our argument.
\begin{lemma}[Schauder estimates and smoothness] \label{L Schauder}
    Fix a point $(t,s) \in [0,T) \times \R^d_+$ and a neighborhood $\mathcal{U}$ of $(t,s)$.
    Suppose Assumption~(A1) holds in conjunction with Inequality~\eqref{E ineq A2} for all $\xi \in \R^d$ and
    $(u,y) \in \mathcal{U}$ and some $C > 0$.
    Let $(f_k)_{k \in \N}$ denote a sequence of solutions of PDE~\eqref{E PDE h} on $\mathcal{U}$,
    uniformly bounded under the supremum norm on $\mathcal{U}$.
    If $\lim_{k \rightarrow \infty} f_k(t,s) = f(t,s)$ on $\mathcal{U}$ for some
    function $f: \mathcal{U} \rightarrow \R$, then $f$ solves PDE~\eqref{E PDE h}
    on some neighborhood $\tilde{\mathcal{U}}$ of $(t,s)$.  In particular, $f \in C^{1,2}(\tilde{\mathcal{U}})$.
\end{lemma}
\begin{proof}
    We refer the reader to the arguments and references provided in \citet{JT}, Section~2 and \citet{ET}, Theorem~3.2.
    The central idea is to use the interior Schauder estimates by \citet{Knerr} along with Arzel\`a-Ascoli type of
    arguments to prove the existence of first- and second-order derivatives of $f$.
\end{proof}
We can now prove the smoothness of the hedging price $h^p$.
\begin{thm}  \label{T diff}
    Under Assumptions~(A1)-(A3) there exists for all points of
    support $(t,s)$ for $S(\cdot)$ some neighborhood $\mathcal{U}$ of $(t,s)$ such that
      the function $h^p$ defined in \eqref{D g} is  in $C^{1,2}(\mathcal{U})$.
\end{thm}
\begin{proof}
    We define $\tilde{p}: \R^{d+1}_+ \rightarrow \R_+$ by $\tilde{p}(s_1,\ldots,s_d,z) := z p(s_1, \ldots, s_d)$ and
    $\tilde{p}^M: \R^{d+1}_+ \rightarrow \R_+$ by
    $\tilde{p}^M(\cdot) := \tilde{p}(\cdot) \1_{\{\tilde{p}(\cdot) \leq M\}}$ for some $M>0$
    and approximate $\tilde{p}^M$ by a sequence of continuous functions $\tilde{p}^{M,m}$
    \citep[compare for example][Appendix~C.4]{Evans} such that
    $\lim_{m \rightarrow \infty} \tilde{p}^{M,m} = \tilde{p}^{M}$ pointwise
    and $\tilde{p}^{M,m} \leq 2M$ for all $m \in \N$.  The corresponding expectations are defined as
    \begin{equation*}
        \tilde{h}^{p,M}(u,y) := \E^{u,y}[\tilde{p}^M(S_1(T), \ldots, S_d(T), \tilde{Z}^{\theta, u,y}(T))]
    \end{equation*}
     for all
    $(u,y) \in \tilde{\mathcal{U}}$ for some neighborhood $\tilde{\mathcal{U}}$ of $(t,s)$
    and equivalently $\tilde{h}^{p,M,m}$.

    We start by proving continuity of $\tilde{h}^{p,M,m}$ for large $m$.
    For any sequence  $(t_k, s_k)_{k \in \N} \subset [0,T] \times \R^d_+$ with
    $\lim_{k \rightarrow \infty} (t_k, s_k) = (t,s)$,
    Lemma~\ref{L flow}, in connection with Assumption~(A1), yields
    $$\lim_{k \rightarrow \infty} \tilde{p}^{M,m}(S^{t_k,s_k}(T), \tilde{Z}^{\theta,t_k,s_k}(T))
    = \tilde{p}^{M,m}(S^{t,s}(T), \tilde{Z}^{\theta,t,s}(T)).$$
    The continuity of $\tilde{h}^{p,M,m}$ follows then from the bounded convergence theorem.

    Now, \citet{JT}, Lemma~2.6, in connection with Assumption~(A2), guarantees that $\tilde{h}^{p,M,m}$
    is a solution of
    PDE~\eqref{E PDE h}. Lemma~\ref{L Schauder} then yields that firstly, $\tilde{h}^{p,M}$ and secondly,
    in connection with Assumption~(A3), $h^p$ also solve PDE~\eqref{E PDE h} on some neighborhood $\mathcal{U}$ of
    $(t,s)$. In particular, $h^p$ is in $C^{1,2}(\mathcal{U})$.
\end{proof}
The last theorem is a generalization of the results in \citet{ET} to several dimensions and to non-continuous payoff
functions $p$.  \citet{Friedman_SDE}, Chapters~6 and 15 and \citet{JT} have related results, but they impose linear
growth conditions on $a(\cdot, \cdot)$ so that PDE~\eqref{E PDE h} has a unique solution of polynomial growth. We are
especially interested in the situation in which multiple solutions may exist. \citet{HS} present results in the case
when the process
corresponding to PDE~\eqref{E PDE h} does not leave the positive orthant.  As \citet{FK} observe,
this condition does not necessarily hold if there is no ELMM.
In the case of $Z^\theta(\cdot)$ being a
martingale, our assumptions are only weakly more general than the ones in \citet{HS}
 by not requiring $a(\cdot, \cdot)$ to be continuous in the time dimension.
However, in all these research articles the authors show that the function $h^p$ indeed solves
PDE~\eqref{E PDE h} not only locally but globally and satisfies the corresponding boundary conditions.
We have here abstained from imposing the stronger assumptions
these papers rely on and concentrate on the local properties of $h^p$. For our application it is sufficient to
observe that
$h^p(t, S(t))$ converges to $p(S(T))$ as $t$ goes to $T$; compare the proof of Theorem~\ref{T Respresentation markovian}.

The next section provides an interpretation of our approach to prove the differentiability of $h^p$;
all problems on the spatial boundary, arising for example from a
discontinuity of $a(\cdot, \cdot)$ on the boundary of the positive orthant, have been ``conditioned away,'' so that $S(\cdot)$ can
get close to but never actually attains the boundary.
\section{CHANGE OF MEASURE} \label{S hedging price}
In order to compute optimal strategies we need to compute the ``deltas'' of
expectations.
To simplify the computations we suggest in this section a change of measure under which
the dynamics of the stock price process simplify.

\citet{DS_existence}, Theorem~1.4 show that NA implies the existence of a local martingale measure
absolutely continuous with respect to $\Prob$. On the other side, a consequence of  this section is the existence
of a local martingale measure under NUPBR, such that $\Prob$ is absolutely continuous with respect to it.
Indeed, NA and NUPBR together yield NFLVR \citep[compare][Proposition~3.2]{DS_fundamental,KK}, which again yields an
ELMM corresponding exactly to the one discussed in this section.  Another point of view,
which we do not take here, is the recent insight by
\citet{Kardaras_finitely} on the equivalence of NUPBR and the existence of a
 finitely additive probability measure which is, in some sense, weakly equivalent to $\Prob$ and under which
$S(\cdot)$ has some notion of weak local martingale property.

Our approach via a
``generalized change of measure'' is in the spirit of the work by \citet{F1972},  \citet{M}, \citet{DS_Bessel},
Section~2, and \citet{FK}, Section~7.
They show that for the  strictly positive $\Prob$-local martingale $Z^\theta(\cdot)$
a probability measure $\Q$ exists such that $\Prob$ is absolutely continuous with
respect to $\Q$ and $d\Prob/d\Q = 1/Z^\theta(T \wedge \tau^\theta)$, where $\tau^\theta$ is the first hitting time of zero by
the process $1/Z^\theta(\cdot)$.
Their analysis has been built upon by several authors, for example by \citet{PP}, Section~2.
We complement this research direction by determining the dynamics of the
$\Prob$-Brownian motion $W(\cdot)$ under the new measure $\Q$.
These dynamics do not follow directly from an
application of a Girsanov-type argument since $\Q$ need not be absolutely continuous with respect to $\Prob$.
Similar results for the dynamics have been obtained in \citet{Sin}, Lemma~4.2 and \citet{DShir}, Section~2.
However, they rely on additional assumptions on the existence of solutions for some stochastic differential equations.
\citet{WH_martingale} prove the existence of a measure $\tilde{\Q}$ satisfying
$\E^\Prob[Z^\theta(T)] = \tilde{\Q}(\tau^\theta > T)$, where $W(\cdot)$ has the same $\tilde{\Q}$-dynamics as we
derive, but $\Prob$ is not necessarily absolutely continuous with respect to $\tilde{\Q}$.

For the results in this section, we make the technical assumption
that the probability space $\Omega$ is the space of
right-continuous paths $\omega: [0, T] \rightarrow \R^m \cup
\{\Delta\}$ for some $m \in \N$
 with left limits at $t \in [0,T]$ if $\omega(t) \neq \Delta$
 and with an absorbing ``cemetery'' point
$\Delta$. By that we mean that $\omega(t) = \Delta$ for some $t
\in [0, T]$ implies $\omega(u) = \Delta$ for all $u \in [t, T]$
and for all $\omega \in \Omega$. This point $\Delta$ will
represent explosions of $Z^\theta(\cdot)$, which do not occur
under $\Prob$, but may occur under a new probability measure $\Q$
constructed below. We further assume that the filtration
$\mathbb{F}$ is the right-continuous modification of the
filtration generated by the paths $\omega$ or, more precisely,  by
the projections $\xi_t(\omega) := \omega(t)$. Concerning the
original probability measure we assume that $\Prob(\omega:
\omega(T) = \Delta) = 0$ and that for all $t \in [0, T]$, $\infty$
is an absorbing state for $Z^\theta(\cdot)$; that is, $Z^\theta(t)
= \infty$ implies $Z^\theta(u) = \infty$ for all $u \in [t, T]$.
This assumption specifies $Z^\theta(\cdot)$ only on a set of
measure zero and is made for notational convenience.

We emphasize that we have not assumed completeness of the filtration $\mathbb{F}$. Indeed, we shall construct
a new probability measure $\Q$ which is not necessarily equivalent to the original measure $\Prob$ and can assign positive probability to
nullsets of $\Prob$. If we had assumed completeness of $\mathbb{F}$, we could not guarantee that $\Q$ could be consistently defined
on all subsets of these nullsets, which had been included in $\mathbb{F}$ during the completion process. The fact that we
need the cemetery point $\Delta$ and cannot restrict ourselves to the original canonical space is also not surprising.
The point $\Delta$ represents events which have under $\Prob$ probability zero, but under $\Q$ have positive probability.

All these assumptions are needed to prove the existence of a measure $\Q$ with
$d\Prob/d\Q = 1/Z^\theta(T \wedge \tau^\theta)$.
After having ensured its existence, one then can take the route suggested by \citet{DS_Bessel}, Theorem~5 and
start from any probability space satisfying the usual conditions, construct a canonical probability space satisfying the
technical assumptions mentioned above, doing all necessary computations on this space, and then going back to the
original space.

For now, the goal is to construct a measure $\Q$ under which the computation of $h^p$ simplifies. For that, we define
the sequence of stopping times
\begin{equation*}
    \tau^\theta_i := \inf\{t \in [0,T]: Z^\theta(t) \geq i\}
\end{equation*}
with $\inf \emptyset := \infty$ and the sequence of
$\sigma$-algebras $\CF^{i} := \CF(\tau^\theta_i \wedge T)$ for all $i \in \N$. We observe that the definition of $\CF^{i}$ is
independent of the probability measure and define the stopping time $\tau^\theta := \lim_{i \rightarrow \infty} \tau^\theta_i$
with corresponding $\sigma$-algebra
$\CF^{\infty, \theta} := \CF(\tau^\theta \wedge T)$ generated by $\cup_{i=1}^\infty \CF^{i,\theta}$.

Within this framework,
\citet{M} and \citet{F1972}, Example~6.2.2 rely on an extension theorem \citep[compare][Chapter~5]{Pa} to show
the existence of a measure $\Q$  on $(\Omega, \CF(T))$ satisfying
\begin{equation} \label{D Q 1}
    \Q(A) = \E^{\Prob}\left[Z^\theta(\tau^\theta_i \wedge T) \1_A\right]
\end{equation}
for all $A \in \CF^{i,\theta}$, where we now write $\E^{\Prob}$ for the expectation under the original measure.
We summarize these insights in the following theorem, which
also generalizes
the well-known Bayes' rule for classical changes of measures \citep[compare][Lemma~3.5.3]{KS1}.
\begin{thm}[Generalized change of measure, Bayes' rule] \label{T Gen change}
    There exists a measure $\Q$ such that $\Prob$ is absolutely continuous with respect to $\Q$ and such that
    for  all
    $\CF\left(T\right)$-measurable random variables $Y \geq 0$ we have
    \begin{align}  \label{E cond exp}
        \E^{\Q}\left[\left.Y \1_{\left\{1/Z^\theta\left(T\right) > 0\right\}}\right|\CF(t)\right] =
        \E^\Prob\left[Z^\theta\left(T\right)
         Y|\CF(t)\right] \frac{1}{Z^\theta\left(t\right)} \1_{\left\{1/Z^\theta\left(t\right) > 0\right\}}
    \end{align}
    $\Q$-almost surely (and thus, $\Prob$-almost surely) for all $t \in [0,T]$,
    where $\E^{\Q}$ denotes the expectation with respect to the new
    measure $\Q$.
    Under this measure $\Q$, the process
    $\widetilde{W}(\cdot) = \left(\widetilde{W}_1(\cdot), \ldots \widetilde{W}_K(\cdot)\right)\tr$ with
    \begin{equation} \label{E tilde W}
        \widetilde{W}_k(t \wedge \tau^\theta) := W_k(t\wedge \tau^\theta) + \int_0^{t\wedge \tau^\theta} \theta_k(u, S(u)) du
    \end{equation}
    for all $k = 1, \ldots, K$ and $t \in [0,T]$  is
    a $K$-dimensional Brownian motion stopped at time $\tau^\theta.$
\end{thm}
\begin{proof}
    The existence of a measure $\Q$ satisfying \eqref{D Q 1} follows as in the discussion above.
    We fix an arbitrary set $B \in \CF(t)$.
    It is sufficient to
    show the statement for $Y = \1_A$ where $A \in \CF\left(T\right)$. We have
    $$A = \left(A \cap \left\{\tau^\theta \leq {T}\right\}\right)
        \cup \bigcup_{i=1}^\infty \left(A \cap \left\{\tau^\theta_{i-1} < {T} \leq \tau^\theta_i \right\}\right).$$

    From the fact that $\tau^\theta \leq {T}$ holds if and only if $1/Z^\theta(T) = 0$ holds,
    from the identity in \eqref{D Q 1}, and from the observation that $\Prob\left(\tau^\theta \leq {T}\right) = 0$,
    we obtain
    \begin{align*}
        \Q\left(A \cap \left\{\frac{1}{Z^\theta\left(T\right) } > 0\right\}\cap B\right) =& \sum_{i=1}^\infty
            \Q\left(A \cap \left\{\tau^\theta_{i-1} <T \leq \tau^\theta_i  \right\}  \cap B  \right) \\
            =& \sum_{i=1}^\infty \E^\Prob\left[Z^\theta(\tau^\theta_i \wedge T)
                \1_{A  \cap \left\{\tau^\theta_{i-1} < {T} \leq \tau^\theta_i \right\}\cap B}
               \right] \\
            =& \E^\Prob\left[Z^\theta \left(T\right) \1_{A \cap B}     \right]\\
            =& \E^\Prob\left[Z^\theta\left(t\right) \E^\Prob\left[\left.Z^\theta\left(T\right) \1_A\right|\CF(t)\right]
                \frac{1}{Z^\theta\left(t\right)}\1_{B}  \right]\\
            =& \E^{\Q}\left[\E^\Prob\left[\left.Z^\theta\left(T\right) \1_A\right|\CF(t)\right]
                \frac{1}{Z^\theta\left(t \right)}
                \1_{\left\{1/Z^\theta\left(t\right) > 0\right\}}
                \1_{B} \right].
    \end{align*}
    Here, the last equality follows as the first ones with $\1_A$ replaced by the random variable inside the last
    expression and $T$ replaced by $t$. This yields \eqref{E cond exp}. The fact that $\Prob$ is absolutely continuous with respect to
    $\Q$ follows from setting $t=0$ in \eqref{E cond exp}.
    From Girsanov's theorem \citep[compare][Theorem~8.1.4]{RY} we obtain that on $\CF^{i,\theta}$
    the process
    $\widetilde{W}(\cdot)$ is  under
    $\Q$ a $K$-dimensional Brownian motion stopped at $\tau_i^\theta \wedge T$.
     Since $\cup_{i=1}^\infty \CF^{i,\theta}$ generates $\CF^{\infty,\theta}$ and forms a $\pi$-system, we get the dynamics of
     \eqref{E tilde W}.
\end{proof}
Thus, an ELMM exists if and only if $\Q(1/Z^\theta(T) > 0) = 1$.
A further consequence of Theorem~\ref{T Gen change} is the fact that
the dynamics of the stock price process and the reciprocal of the SDF
simplify under $\Q$ as the next corollary shows.
\begin{cor}[Evolution of important processes under $\Q$] \label{C Evolution}
    The stock price process $S(\cdot)$ and the reciprocal $1/Z^\theta(\cdot)$ of the SDF
    evolve until the stopping time $\tau^\theta$ under $\Q$ according to
    \begin{align*}
        d S_i(t) &= S_i(t) \sum_{k=1}^K \sigma_{i,k}(t, S(t)) d \widetilde{W}_k(t),\\
        d \left(\frac{1}{Z^\theta(t)}\right) &=\frac{1}{Z^\theta(t)} \sum_{k=1}^d \theta_k(t,S(t)) d \widetilde{W}_k(t)
    \end{align*}
    for all $i = 1, \ldots, d$ and $t \in [0,T]$.  Furthermore,
    for any process $N(\cdot)$, $N(\cdot) \1_{\{1/Z^\theta(\cdot) > 0\}}$ is a $\Q$-martingale
    if and only if $N(\cdot)Z^\theta(\cdot)$ is a $\Prob$-martingale.
    In particular, the process  $1/Z^\theta(\cdot)$ is a $\Q$-martingale.
\end{cor}
\begin{proof}
    The dynamics are a direct consequence of the representation of $\widetilde{W}(\cdot)$ in \eqref{E tilde W} and the
    definition of the MPR.
    The other statements follow from choosing $Y = N(T)$ and
    $Y = 1/Z^\theta(T)$ in \eqref{E cond exp}.
\end{proof}
The results of the last corollary play an essential role when we do computations, since the first hitting time of
the reciprocal of the SDF
can in most cases be easily represented as a first hitting time of the stock price. This now usually follows  some more
tractable dynamics, as we shall see in Section~\ref{S examples}.
For the case of strict local martingales the equivalence of the last corollary
     is generally not true. Take as an example $N(\cdot) \equiv 1$ and $Z^\theta(\cdot)$ a strict local
    martingale under $\Prob$. Then, $Z^\theta(\cdot) N(\cdot) \equiv Z^\theta(\cdot)$ is a local $\Prob$-martingale but
    $N(\cdot) \1_{\{1/Z^\theta(\cdot) > 0\}} \equiv \1_{\{1/Z^\theta(\cdot) > 0\}}$
    is clearly not a local $\Q$-martingale. The reason for this lack of symmetry is
    that a sequence of stopping times which converges $\Prob$-almost surely to $T$ need not necessarily
    converge $\Q$-almost surely to $T$.
\section{EXAMPLES} \label{S examples}
In this section, we discuss several examples for markets which imply arbitrage opportunities. Examples~\ref{Bes 1} and
\ref{Bes 2}
treat the case of a three-dimensional Bessel process with drift for various payoffs. Example~\ref{Inverse Bessel}
concentrates on the reciprocal of the three-dimensional Bessel, a standard example in the bubbles literature.
\begin{ex}[Three-dimensional Bessel process with drift - money market] \label{Bes 1}
    One of the best known examples for markets without an ELMM is the three-dimensional Bessel process, as
    discussed in \citet{KS1}, Section~3.3.C.
    We study here a class of models which contain the Bessel process as special
    case and generalize the example for arbitrage of A.V.~Skorohod in \citet{KS2}, Section~1.4.
    For that, we begin with defining an auxiliary stochastic process $X(\cdot)$ as a Bessel process with drift $-c$, that is,
    \begin{equation} \label{E Bessel X}
            d X(t) = \left(\frac{1}{X(t)} - c\right) dt + d W(t)
    \end{equation}
    for all $t \in [0,T]$
    with $W(\cdot)$
     denoting a Brownian motion on its natural filtration $\mathbb{F} = \mathbb{F}^W$ and $c \in [0,\infty)$ a constant.
    The process $X(\cdot)$ is  strictly positive, since it is a Bessel process, thus strictly positive under the
    equivalent measure where $\{W(t) - ct\}_{0 \leq t \leq T}$ is a Brownian motion.
    The stock price process is now defined via the stochastic differential equation
    \begin{equation}  \label{E Bessel S}
        d S(t) = \frac{1}{X(t)}dt + d W(t)
    \end{equation}
    for all $t \in [0,T]$.
    Both processes $X(\cdot)$ and $S(\cdot)$ are assumed to start at the same point $S(0)> 0$.
    From \eqref{E Bessel X} and \eqref{E Bessel S} we obtain directly $S(t) = X(t) + ct > 0$
    for all $t \in [0,T]$.
    If
    $c=0$ then $S(\cdot) \equiv X(\cdot)$ and the stock price process is a Bessel process. Of course, the MPR is exactly
    $ \theta(t,s) = 1/ (s - c t)$
    for all $(t,s) \in [0,T] \times \R_+$ with $s>ct$.
    Thus, the reciprocal $1/Z^\theta(\cdot)$ of the SDF
    hits zero exactly when $S(t)$ hits $ct$. This follows directly from the $\Q$-dynamics of $1/Z^\theta(\cdot)$
    derived in Corollary~\ref{C Evolution} and
    a strong law of large numbers as in
    \citet{K_balance}, Lemma~A.2.

    Let us start by looking at a general, for the moment not-specified payoff function $p$.
    For all $(t,s) \in [0,T] \times \R_+$ with $s>ct$, by relying on Theorem~\ref{T Gen change},
    using the density of a Brownian motion absorbed at zero
    \citep[compare][Problem~2.8.6]{KS1} and some simple computations, we obtain
    \begin{align}
        h^p(t,s) =&  \E^{t,s}\left[\tilde{Z}^{\theta, t,s}(T) p(S(T))\right] \nonumber \\
            =&  \left.\E^{\Q}\left[\left.p(S(T)) \1_{\{\min_{t\leq u\leq T} \{S(u) - cu\} > 0\}}
                \right|\CF(t)\right]\right|_{S(t)=s} \nonumber \\
            =&   \int_{\frac{cT-s}{\sqrt{T-t}}}^\infty  \frac{1}{\sqrt{2 \pi}}
                \exp\left(-\frac{z^2}{2} \right) p(z\sqrt{T-t}   + s) dz \nonumber \\
            &-  \exp(2 cs - 2c^2t) \int_{\frac{cT - 2ct+s}{\sqrt{T-t}}}^\infty \frac{1}{\sqrt{2 \pi}}
                \exp\left(-\frac{z^2}{2} \right) p(z\sqrt{T-t}   - s + 2ct) dz.  \label{E Bessel U f}
    \end{align}

    Let us consider the investment in the money market only, to wit, $p(s) \equiv p^0(s) \equiv 1$
    for all  $s > 0$. The expression in \eqref{E Bessel U f} yields the hedging price of one monetary unit
    \begin{equation}   \label{E Bessel U bond}
        h^{p^0}(t,s) = \Phi\left(\frac{s-cT}{\sqrt{T-t}}\right)
            - \exp(2 cs - 2 c^2t) \Phi\left(\frac{-s-cT+2ct}{\sqrt{T-t}}\right),
    \end{equation}
    where $\Phi$ denotes the cumulative standard normal distribution function.
    It can be easily checked that $h^{p^0}$ solves PDE~\eqref{E PDE h} for all
    $(t,s) \in [0,T] \times \R_+$ with $s>ct$. Thus, by Theorem~\ref{T Respresentation markovian}
    the optimal hedging strategy $\eta^0$ of one monetary unit is
    \begin{equation*}
        \eta^0(t, s) = \frac{2}{\sqrt{T-t}} \phi\left(\frac{s-cT}{\sqrt{T-t}}\right) - 2c \exp(2 cs - 2 c^2t) \Phi\left(\frac{-s-cT + 2ct}{\sqrt{T-t}}\right),
    \end{equation*}
    where $\phi$ denotes the standard normal density.

    It is well-known that a Bessel process allows for arbitrage. Compare for example \citet{KK},
    Example~3.6 for an ad-hoc strategy which corresponds to a hedging price of $\Phi(1)$ for a monetary unit if
    $c = 0$ and $S(0)=T=1$.
    We have improved here the existing strategies and found the optimal one, which corresponds in this setup
    to a hedging price of $h^{p^0}(0,1) = 2\Phi(1) - 1 < \Phi(1)$.
\end{ex}
\begin{remark}[Multiple solutions for PDE~\eqref{E PDE h}]
    We observe that the hedging price $h^{p^0}$ in \eqref{E Bessel U bond} depends on the drift $c$.
    Also, $h^{p^0}$ is sufficiently differentiable, thus by Remark~\ref{R non-uniqueness} uniquely characterized
    as the minimal nonnegative solution of PDE~\eqref{E PDE h}, which does not depend on the drift $c$.
    The uniqueness of  $h^{p^0}$ by Remark~\ref{R non-uniqueness} and the dependence of $h^{p^0}$ on $c$
    do not contradict each other, since
    the nonnegativity of  $h^{p^0}$ has only to hold at the points of support for $S(\cdot)$.
    For a given time $t\in [0,T]$, these are
    only the points $s>ct$. Thus, as $c$ increases, the nonnegativity condition weakens since it has to hold for fewer
    points, and thus $h^p$ can become smaller and smaller. Indeed, plugging in \eqref{E Bessel U bond}
    the point $s=ct$ yields $h^{p^0}(t,ct) = 0$.  In summary, while the PDE itself does only depend on the
    (more easily observable) volatility structure of the stock price dynamics, the mean rate of return determines where the PDE
    has to hold.
\end{remark}
In the next example we price and hedge a European call within the same class of models as in the last example.
\begin{ex}[Three-dimensional Bessel process with drift - stock and European call] \label{Bes 2}
Plugging in \eqref{E Bessel U f} the payoff $p(s) = p^C(s) = (y-L)^+$ for some $L \geq 0$
 and writing $\tilde{L} := \max\{cT,L\}$,
a simple computation yields
\begin{align*}
    h^{p^C}(t,s)
        =& \sqrt{\frac{T-t}{2\pi}} \exp\left(-\frac{(s-\tilde{L})^2}{2 (T-t)}\right) +
            (s-L) \Phi\left(\frac{s-\tilde{L}}{\sqrt{T-t}}\right)
        - \exp(2 cs - 2c^2t) \\
        &\cdot \left(\sqrt{\frac{T-t}{2\pi}}
            \exp\left(-\frac{(\tilde{L} - 2c t + s)^2}{2 (T-t)}\right)
        + (2ct - s - L) \Phi\left(\frac{-\tilde{L}+2ct-s}{\sqrt{T-t}}\right)\right).
    \end{align*}
If $L\leq cT$, in particular if $L=0$, the last expression simplifies to
\begin{equation*}
    h^{p^C}(t,s)
        = s \Phi\left(\frac{s-cT}{\sqrt{T-t}}\right) + \exp(2 cs - 2c^2t)
            \Phi\left(\frac{2ct-s-cT}{\sqrt{T-t}}\right)
            (s - 2ct)
           -L  h^{p^0}(t,s),
\end{equation*}
where $h^{p^0}$ denotes the hedging price of one monetary unit given in \eqref{E Bessel U bond}.
It is simply the difference between the hedging price of the stock and $L$ monetary units since if
$L\leq cT$, the call is always exercised.
Using $L=0$ we get the value of the stock.  We could now proceed by computing the derivative of $h^{p^C}$ in $s$
to get the hedge.  Furthermore, the modified put-call parity of Corollary~\ref{C put call} provides us
directly with the hedging price for a put.

If $L=c = 0$, we write $p^1 \equiv p^L$ and the last equality yields $ h^{p^1}(t,s) = s$ for all $(t,s) \in [0,T] \times \R_+$
and holding the stock is optimal. There are two other ways to see this result right away. Simple computations
show directly that $\tilde{Z}^{\theta,t,s}(T) = s/S(T)$ if $c = 0$, thus
$h^{p^1}(t,s) = s$ for all $(t,s) \in [0,T] \times \R_+$. Alternatively,
using the representation of $h^{p^1}(t,s)$ implied by \eqref{E cond exp} we see that the hedging
price is just the expectation of a Brownian motion stopped at zero, thus the expectation of a martingale started at $s$.

Two notable observations can be made. First, in this model both the money market and the stock
simultaneously have a hedging price cheaper than their current price, as long as $c>0$.
Second, in contrast to classical theory, the mean rate of return under the
``real-world'' measure does matter in determining the hedging price of calls (or other derivatives).
\end{ex}
\citet{PP} compute call prices for the reciprocal Bessel process model.
We discuss next how the results
of the last examples relate to this model.
\begin{ex}[Reciprocal of the three-dimensional Bessel process] \label{Inverse Bessel}
    Let the stock price $\tilde{S}(\cdot)$ have the dynamics
    \begin{equation*}
            d \tilde{S}(t) = - \tilde{S}^2(t) d W(t)
    \end{equation*}
    for all $t \in [0,T]$ with $W(\cdot)$ denoting a Brownian motion on its natural filtration $\mathbb{F} = \mathbb{F}^W$.
    The process $\tilde{S}(\cdot)$ is exactly the reciprocal of the process $S(\cdot)$ of Examples~\ref{Bes 1} and \ref{Bes 2}
    with $c=0$, thus strictly positive.
    We observe that  $\Prob$ is already
    a martingale measure. However, if one
    wants to hold the stock at time $T$, one should not buy the stock at time zero, but use the strategy
    $\eta^1$ below for a hedging price smaller than $\tilde{S}(0)$
    along with the suboptimal strategy $\eta(\cdot, \cdot)  \equiv 1$. That is, the stock has a bubble.

    We have already observed that $\tilde{S}(T) = 1/S(T)$, which is exactly the SDF in
    Example~\ref{Bes 1} for $c=0$ multiplied by $\tilde{S}(t)$. Thus, as in
    \eqref{E Bessel U bond} with $c=0$, the hedging price for the stock is
    \begin{equation} \label{E h price iB}
        h^{p^1}(t,s) = 2 s \Phi\left(\frac{1}{s \sqrt{T-t}}\right) - s < s
    \end{equation}
    along with the optimal strategy
    \begin{equation*}
        \eta^1 (t, s) = 2 \Phi\left(\frac{1}{s \sqrt{T-t}}\right) - 1 - \frac{2}{s\sqrt{T-t}}   \phi\left(\frac{1}{s\sqrt{T-t}}\right)
    \end{equation*}
    for all $(t,s) \in [0,T) \times \R_+$.
    For pricing calls, we observe
    \begin{equation*}
        \left(\tilde{S}(T) - L\right)^+ = L \tilde{S}(T) \left(\frac{1}{L} - \frac{1}{\tilde{S}(T)}\right)^+
            = \frac{L}{S(t)} \cdot \frac{S(t)}{S(T)} \left(\frac{1}{L} - S(T)\right)^+
    \end{equation*}
    for $L > 0$. Thus, the price at time $t$ of a call with strike $L$ in the reciprocal Bessel model is the price of
    $L\tilde{S}(t)$
    puts with strike $1/L$ in the Bessel model  and can be computed from Example~\ref{Bes 2} and
    Corollary~\ref{C put call}. For $S(0)=1$, simple computations will lead directly to Equation~(6) of \citet{PP}.
    The optimal strategy could now be derived with Theorem~\ref{T Respresentation markovian}.
\end{ex}
\section{CONCLUSION} \label{S conclusion}
It has been proven that, under weak technical assumptions, there is no equivalent local martingale measure needed
to find an optimal hedging strategy based upon the familiar delta hedge. To ensure its existence, weak sufficient
conditions have been introduced which guarantee the differentiability of an expectation parameterized over time and
over the original market configuration.
The dynamics of stochastic processes simplify after
 a non-equivalent change of  measure and a generalized Bayes' rule has been derived.
With this newly developed
machinery, some optimal trading strategies have been computed addressing standard examples for which so far only ad-hoc
and not necessarily optimal strategies have been known.

\bibliographystyle{apalike}
\setlength{\bibsep}{1pt}
\small\bibliography{aa_bib}{}
\end{document}